\documentclass[twocolumn]{revtex4-1}
\usepackage{amsmath}
\usepackage{amsthm}
\usepackage{latexsym}
\usepackage{amsfonts}
\usepackage{amssymb}
\usepackage{color}
\usepackage{bbm,dsfont}
\usepackage{graphicx}
\usepackage{MnSymbol}
\usepackage{enumerate}
\usepackage{bbding}

\newtheorem{proposition}{Proposition}

\newtheorem{lemma}{Lemma}
\newtheorem{corollary}{Corollary}

\newtheorem*{question*}{Question}

\newcommand{\R}{\mathbb{R}} 
\newcommand{\C}{\mathbb{C}} 
\newcommand{\real}{\mathbb R} 
\newcommand{\complex}{\mathbb C}

\newcommand{\e}{{\rm e}} 

\newcommand{\no}[1]{\left\|#1\right\|} 
\newcommand{\tr}[1]{{\rm tr}\left[#1\right]} 

\newcommand{\ran}{\textrm{ran}} 
\newcommand{\id}{\mathbbm{1}} 
\newcommand{\flip}{\mathbbm{F}} 

\renewcommand{\rho}{\varrho}
\newcommand{\lam}{\lambda}

\renewcommand{\Re}{{\rm Re}\,}
\renewcommand{\Im}{{\rm Im}\,}

\newcommand{\dett}{{\rm det}} 

\begin{document}\setlength{\arraycolsep}{2pt}

\title[]{Verifying the quantumness of bipartite correlations}

\author{Claudio Carmeli}
\email{claudio.carmeli@gmail.com}
\affiliation{DIME, Universit\`a di Genova, Via Magliotto 2, I-17100 Savona, Italy}

\author{Teiko Heinosaari}
\email{teiko.heinosaari@utu.fi}
\affiliation{Turku Centre for Quantum Physics, Department of Physics and Astronomy, University of Turku, FI-20014 Turku, Finland}

\author{Antti Karlsson}
\email{aspkar@utu.fi}
\affiliation{Turku Centre for Quantum Physics, Department of Physics and Astronomy, University of Turku, FI-20014 Turku, Finland}

\author{Jussi Schultz}
\email{jussi.schultz@gmail.com}
\affiliation{Turku Centre for Quantum Physics, Department of Physics and Astronomy, University of Turku, FI-20014 Turku, Finland}

\author{Alessandro Toigo}
\email{alessandro.toigo@polimi.it}
\affiliation{Dipartimento di Matematica, Politecnico di Milano, Piazza Leonardo da Vinci 32, I-20133 Milano, Italy}
\affiliation{I.N.F.N., Sezione di Milano, Via Celoria 16, I-20133 Milano, Italy}

\begin{abstract}
Entanglement is at the heart of most quantum information tasks, and therefore considerable effort has been made to find methods of  deciding the entanglement content of a given bipartite quantum state. Here, we prove a fundamental limitation to deciding if an unknown state is entangled or not: we show that any quantum measurement which can answer this question necessarily gives enough information to identify the state completely. Therefore, only prior information regarding the state can make entanglement detection less expensive than full state tomography in terms of the demanded quantum resources. We also extend our treatment to other classes of correlated states by considering the problem of deciding if a state is NPT, discordant, or fully classically correlated. Remarkably, only the question related to quantum discord can be answered without resorting to full state tomography.

\end{abstract}

\maketitle


The advent of quantum information theory has brought entanglement from  a peculiarity of quantum theory into a genuine resource which allows the performance of tasks that are beyond what is possible within classical physics. The vast sea of applications of entanglement include super dense coding \cite{BeWi92}, teleportation \cite{BeBrCrJoPeWo93} and one-way quantum computation \cite{RaBr01}, to name a few. Due to the significant role of entanglement in these applications, a considerable effort has been made to find methods of verifying if a given bipartite system is entangled or not \cite{HoHoHoHo09, GuTo09}.  Many of the existing methods such as entanglement witnesses \cite{HoHoHo96, ChSa14} or Bell inequalitites \cite{Bell64, ClHoShHo69} are able to detect the entanglement of specific types of states, and it seems that a detour via full state tomography is the only way to give a definite answer for an arbitrary unknown state. However, even with complete knowledge of the quantum state, the problem is computationally extremely difficult \cite{Gu03,Io07}, and since deciding if a state is entangled or not is a simple yes-no question, one might hope to do this with less resources. In this article, we address the problem of entanglement detection from a very fundamental point of view by formulating and answering the following question: 

\begin{quote}
\emph{What is the minimum amount of information required from a quantum measurement to allow us to infer whether or not an unknown state is entangled?}
\end{quote} 

In order to properly motivate this question, we begin with an illustrative example of the corresponding problem for \emph{pure} states.
Suppose we have a bipartite system consisting of two $d$-dimensional subsystems. 
The system is in an unknown pure state described by a unit vector $\vert\psi\rangle$ in the $D=d^2$-dimensional Hilbert space of the composite system, and it is entangled if it cannot be decomposed as a factorized vector $\vert\psi\rangle = \vert\phi\rangle\otimes \vert\varphi\rangle$. One possible way of deciding if this is the case, is to first identify completely which pure state the system is in and then look at its Schmidt decomposition. By exploiting prior knowledge about the purity of the state this can be done with a measurement consisting of $\sim 4D=4d^2$ different outcomes \cite{HeMaWo13}. However, since  a pure state is entangled if and only if its reduced states are mixed, it is sufficient to perform a local measurement  which answers this question on either subsystem. This can be achieved with just $\sim 5d$ outcomes \cite{ChDaJietal13}, which shows a significant reduction in the resources needed for answering the question. 

The essential conclusion in the above example is that it is possible to find out if a pure state is entangled \emph{without} identifying the state. In fact, a direct comparison of the minimal number of outcomes $4d^2$ and $5d$ shows that the latter measurement is not sufficient, even in principle, for identifying an unknown pure state. For mixed states the situation changes drastically. A general mixed state of a bipartite system is described by a $D\times D$ density matrix $\varrho$, and it is  entangled if it does not admit a convex mixture $\varrho = \sum_{i} c_i \sigma_i\otimes \eta_i$
in terms of factorized states. Clearly, no strategy involving only local measurements on one party can be useful, but our main result says even more, namely, that:

\begin{quote}
\emph{Any measurement which can be used to decide if a bipartite system is in an entangled state is necessarily giving enough information to identify the unknown state completely}.
\end{quote}

As we will explain, this fact can be formulated as a geometric property of the subset of entangled states.
The geometric approach can also be applied to other subsets of a bipartite state space, and we study the question for NPT states, discordant states, and fully classically correlated states. 
Interestingly, only the membership to the subset of discordant states can be determined without  full knowledge of the state. For two qubits, this result is implicit in \cite{DaVeBr10}.

\section*{Verifying  properties of quantum states}\label{sec:problem}
Entanglement detection can be thought of as a membership problem: it is a question about whether or not an unknown state belongs to the set of entangled states. The same question can obviously be asked for any subset $\mathcal{P}$ of the state space $\mathcal{S}$ of the system. Physically one can think of this as  a problem of deciding if the system possesses some predefined property represented by $\mathcal{P}$. In order to answer this, one must  perform a measurement on the system, represented by a \emph{positive operator valued measure} (POVM) \cite{MLQT12}. A POVM with a finite number of outcomes is a map $E$ that assigns a positive operator $E_j$ to each measurement outcome $j$ and satisfies the normalization $\sum_j E_j = \id$.
The measurement of  $E$ in a state $\varrho$ then yields measurement outcomes which are distributed according to the probabilities $p_j=\tr{\varrho E_j}$.

If the measurement of $E$ is meant to give an answer to the membership problem, then the states in $\mathcal{P}$ must be distinguishable from the states in the complement $\mathcal{P}^c = \mathcal{S}\setminus\mathcal{P}$ solely in terms of the measurement outcome statistics. In order to get a grasp of what this means, notice that  if two states $\varrho_1$ and $\varrho_2$ give the same statistics, then  $\tr{(\varrho_1-\varrho_2)E_j}=0$ for all $j$. In other words,  the operator $\varrho_1-\varrho_2$ belongs to the subspace 
\begin{equation*}
\mathcal{X}_E=\{ \Delta \mid \tr{\Delta E_j}=0\,\text{ for all }\, j \}
\end{equation*}
of traceless Hermitian operators which are orthogonal to all the POVM elements $E_j$. The geometric idea behind this is depicted in Fig.~\ref{fig:geometry}. By turning this around, we see that a measurement of $E$ determines if a state belongs to $\mathcal{P}$, if and only if the subspace $\mathcal{X}_E$ has the property that no nonzero operator $\Delta\in\mathcal{X}_E$  can be decomposed as  $\Delta=\lambda(\varrho_1-\varrho_2)$ with $\varrho_1\in\mathcal{P}$ and $\varrho_2\in\mathcal{P}^c$ and some  scalar $\lambda$.

\begin{figure}
\centering
            \includegraphics[width=8cm]{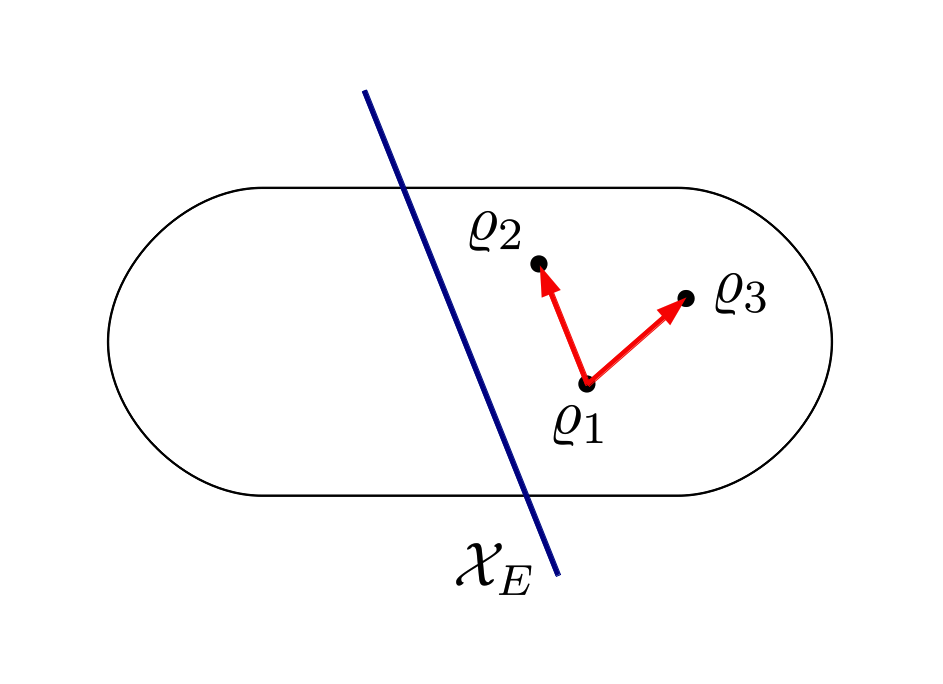} 
           \caption{Geometric idea behind state distinction. The state space and the space $\mathcal{X}_E$ orthogonal to all POVM elements $E_j$ are embedded in a common space. The states $\varrho_1$ and $\varrho_2$ cannot be distinguished by measuring $E$, since the line segment from $\varrho_1$ to $\varrho_2$ is parallel to $\mathcal{X}_E$. For the states $\varrho_1$ and $\varrho_3$ this is not the case, and therefore a measurement of $E$ can be used to distinguish between them.}
    \label{fig:geometry}
\end{figure}

In the extreme case that $\mathcal{X}_E = \{ 0\}$, the POVM $E$ is called \emph{informationally complete} \cite{Prugovecki77}, and it means that any two states are distinguishable from the statistics. Such measurements form the basis of quantum state tomography, and always provide a trivial answer to the membership problem. However, an informationally complete POVM always requires $D^2$ different measurement outcomes, and thus it is the most demanding in terms of quantum resources. Our aim instead is to find optimal measurements which minimize the number of  needed outcomes. To this end, notice that the dimension of the subspace $\mathcal{X}_E$ depends only on the number of linearly independent POVM elements and as such, on the minimal number of outcomes for the optimal measurement. In fact, if we denote the number of linearly independent elements by $\dim E$,  then we have the equality
\begin{equation}\label{eqn:dim} 
D^2 = \dim E + \dim \mathcal{X}_E 
\end{equation}
as the real linear space of all Hermitian operators splits into a direct sum of $\mathcal{X}_E$ and the span of $E$. This means that in our search for an optimal measurement for verifying a given property $\mathcal{P}$, we can  look for the largest subspace $\mathcal{X}_E$ which still has the property stated in the previous paragraph.

\section*{Geometrical viewpoint}\label{sec:geometry}

When the problem of verifying a property of a quantum state is viewed as a membership problem, it allows an intuitive geometrical interpretation. Indeed, using the generalized Bloch representation $\varrho = \frac{1}{D}(\id + \vec{r}\cdot \vec{\sigma})$, we can view the state space of a $D$-dimensional quantum system as  a subset of $\real^{D^2-1}$ \cite{GQS06}. 
Moreover, since the matrices $\{ \sigma_j \}_{j=1}^{D^2-1}$ form a basis of traceless Hermitian matrices, we can embed also $\mathcal{X}_E$ into the same space.

For a given property $\mathcal{P}$ (now viewed also as a subset of $\real^{D^2-1}$) and a POVM $E$, the relevant question was if any $\Delta\in\mathcal{X}_E$ could be decomposed as $\Delta = \lambda(\rho_1-\rho_2)$ with $\rho_1\in\mathcal{P}$ and $\rho_2 \in\mathcal{P}^c$. Solving for $\rho_1 = \rho_2 + \lambda^{-1}\Delta$, we arrive at the following observation: a POVM $E$ can not be used to verify a property $\mathcal{P}$ if it is possible to draw a line segment from $\mathcal{P}$ to $\mathcal{P}^c$ which is parallel to $\mathcal{X}_E$. Conversely, $E$ can be used for the task if there does not exist a line parallel to $\mathcal{X}_E$ which intersects both $\mathcal{P}$ and $\mathcal{P}^c$.

Since the state space $\mathcal{S}$ is a compact convex subset of $\real^{D^2-1}$,  we immediately notice that the problem becomes trivial if either $\mathcal{P}$ or $\mathcal{P}^c$ is contained in the interior of $\mathcal{S}$. In that case for any direction one can always find a parallel line which intersects both sets. The physical implication of this is that in order to verify the corresponding property of the state, it is necessary to measure an informationallly complete POVM, see Fig.~\ref{fig:examples}. However, for the properties  we are interested in, such as entanglement, both sets contain states from the boundary of the state space (e.g., pure states can be entangled or separable), and the problem  becomes more intricate.

\begin{figure}
\centering
            \includegraphics[width=9cm]{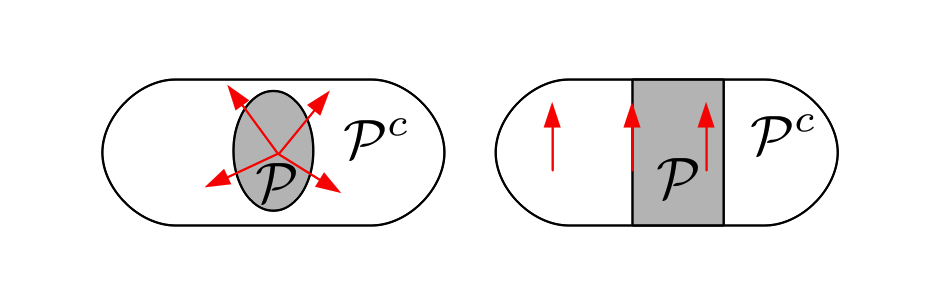} 
           \caption{Two examples of properties to be verified. In the left picture, the set $\mathcal{P}$ is contained in the interior of the state space, and therefore for any direction it is possible to draw a line segment from $\mathcal{P}$ to $\mathcal{P}^c$ in that direction. This means that an informationally complete measurement is needed to decide whether a state is in $\mathcal{P}$ or $\mathcal{P}^c$. In the right picture, the vertical line segments do not cross the boundary between  $\mathcal{P}$ and $\mathcal{P}^c$, and therefore distinguishing between these sets does not require informational completeness.}
    \label{fig:examples}
\end{figure}

Whenever the property is such that it can be verified without informational completeness, we also have interesting geometrical implications regarding the corresponding sets in $\real^{D^2-1}$. Indeed, this means that there exists at least one preferred direction  such that no line parallel to it crosses the boundary between $\mathcal{P}$ and  $\mathcal{P}^c$. 
Hence, the boundary is flat in that direction; see Fig.~\ref{fig:examples}. 
Similarly, if there are several preferred directions which corresponds to one being able to verify the property $\mathcal{P}$ with fewer POVM elements (measurement outcomes), then the boundary is flat in all of these directions.

\section*{Bipartite correlations}\label{sec:correlations}
We now make these general ideas more concrete by specifying the types of properties that we wish to verify. In our case, these are given by different levels of correlations between two quantum systems, see Fig.~\ref{fig:state_space}. 

\begin{figure}
\centering
            \includegraphics[width=9cm]{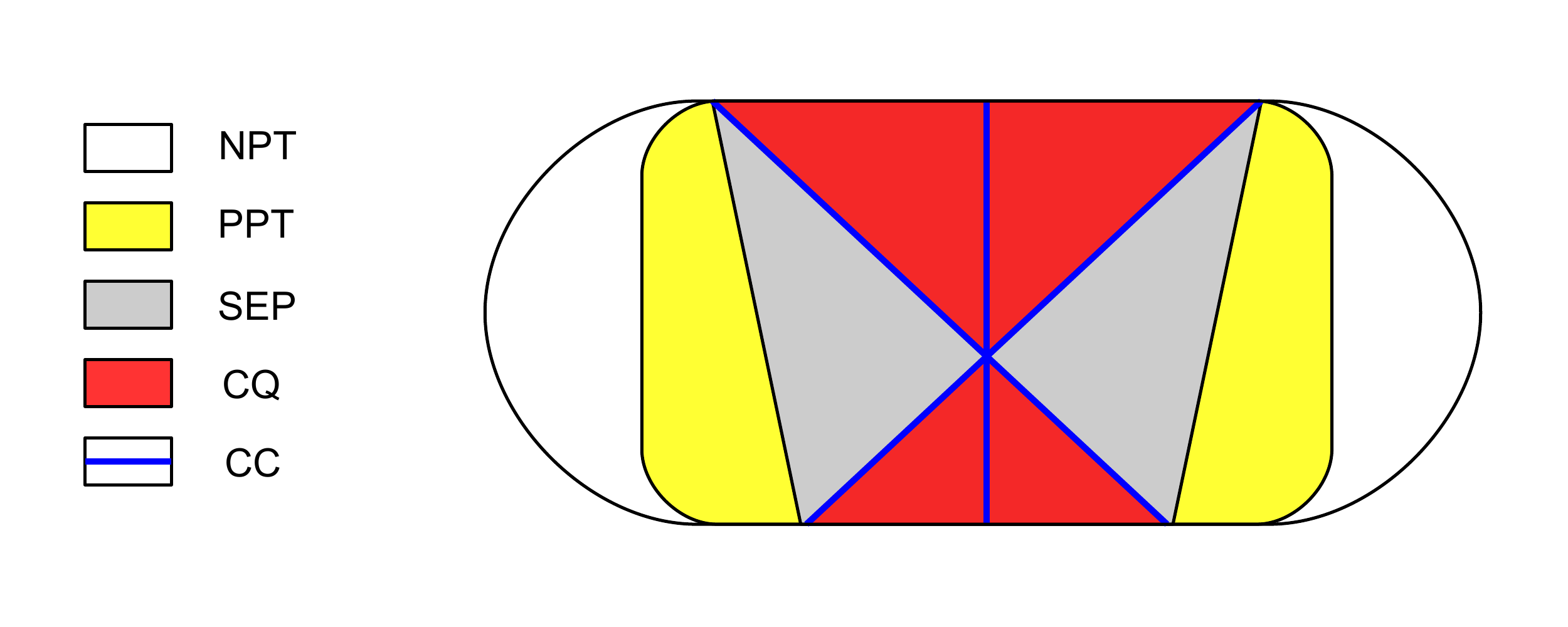} 
           \caption{Schematic of the state space of a bipartite system with the various subsets of states possessing different types of correlations. For the colored areas we have the inclusions ${\rm CC} \subset {\rm CQ} \subset {\rm SEP} \subset {\rm PPT}$, whereas the the NPT states form the complement of the colored region.}
    \label{fig:state_space}
\end{figure}

The most commonly encountered class of bipartite states possessing quantum correlations are the \emph{entangled} states. These are the states which cannot be represented as convex mixtures 
$$
\varrho = \sum_i c_i \sigma_i\otimes \eta_i
$$ 
of product states \cite{Werner89}. The complement of the set of entangled states is the set of separable (SEP) states. The significance of entanglement in quantum information processing is undisputed, and hence the detection of entanglement is of utmost importance. However, even if the state of the bipartite system is known, it is generally an NP-hard problem to decide if it is separable or not \cite{Gu03}. For this reason one typically relies on some more easily manageable sufficient test for entanglement. One such test is based on the NPT (negative partial transpose) criterion: if the partial transpose of $\varrho$ has a negative eigenvalue, then the state is entangled. The converse implication holds only for systems consisting of two qubits or a qubit and a qutrit \cite{HoHoHo96}. The NPT states are also important in their own right, in particular due to their connection with entanglement distillation \cite{HoHoHoHo09}. The division of the state space into NPT and PPT (positive partial transpose) states gives example of a property whose verification we consider.

Even a separable state may have some correlations that are not purely classical in nature. Quantum discord, a certain functional on the state space, has been introduced as a quantifier of the quantumness of these correlations \cite{OlZu01}. 
It has been shown that nonzero discord yields advantage in tasks such as distributing entanglement \cite{CuVeDuCi03}, phase estimation \cite{Gi-etal14} and remote state preparation \cite{Dakic-etal12}, and it allows locking of correlations without entanglement \cite{DiHoLeSmTe04} and  special encoding of information \cite{Gu-etal12}. 
A state has nonzero quantum discord if it cannot be written as 
$$
\varrho = \sum_i c_i \vert \phi_i \rangle \langle \phi_i \vert \otimes \eta_i
$$
for some orthonormal basis $\{\vert \phi_i \rangle\}_{i=1}^d$. If such a decomposition exists, we say that the state is classical-quantum (CQ) in order to emphasize the asymmetric nature of the zero discord states. Similarly, by  exchanging  the local parties we get quantum-classical (QC) states. The states lying in the intersection of these sets are then called classical-classical (CC) and these are the states which can be written as
\begin{equation*}
\varrho = \sum_{ij} c_{ij} \vert \phi_i \rangle\langle \phi_i \vert \otimes \vert \varphi_j \rangle \langle \varphi_j \vert
\end{equation*}
for some orthonormal bases $\{\vert \phi_i \rangle\}_{i=1}^d$ and $\{\vert \varphi_j \rangle\}_{j=1}^d$. The CC states are often called (fully) classically correlated states, and we will use the term non-classical to mean the complement of these states.

\section*{Main results}\label{sec:results}

Taking the general formulation of the verification problem as the starting point and following the geometric intuition outlined above, it is possible to study the problem of verifying the different levels of  correlations in bipartite states. Our main result is that the only question that can be answered without an informationally complete measurement is whether or not the state is CQ. In all of the other cases the measurement necessarily gives enough information to completely identify the state (see Table~\ref{tab:summary}). This may seem quite surprising, in particular, after one notices that the subsets satisfy the following chain of inclusions:
$$
{\rm CC} \subset {\rm CQ} \subset {\rm SEP} \subset {\rm PPT}.
$$
This serves as a demonstration of an important observation regarding these verification problems in general: the size of the subset is irrelevant, but the minimal resources are intimately connected to the \emph{geometry} of the set.

\begin{table}[h]
\centering
\begin{tabular}{| l | c | c |}
\hline
Property   & Informational  &  Minimal number  \\
 to be verified & completeness & of outcomes \\
\hline\hline
NPT & \CheckmarkBold &  $D^2$\\
ENTANGLED & \CheckmarkBold &  $D^2$\\
DISCORDANT  & \XSolidBrush & $D^2-D+1$\\
NON-CLASSICAL & \CheckmarkBold & $D^2$\\
 \hline
\end{tabular}
\caption{Summary of the main results. For the various properties of bipartite states, we indicate whether or not their verification requires an informationally complete measurement, as well as the minimal number of measurement outcomes needed for succeeding in the task. $D=d^2$ is the dimension of the bipartite system's Hilbert space.}
\label{tab:summary}
\end{table}

We now outline the main idea behind the proofs of our results. The detailed mathematical calculations are given in the Supplementary Material.

We begin by looking at the problem in the case of entanglement. Recall that in order to show that entanglement verification requires an informationally complete measurement, we need to show that for any traceless Hermitian operator $\Delta$ we can find a separable state $\varrho_2$ and a real number $\lambda$ such that $\varrho_1  = \varrho_2 + \lambda^{-1} \Delta$ is an entangled state. It is convenient to seek for a suitable $\varrho_2$ from the boundary of the set of separable states, in which case one can think of  $\lambda^{-1} \Delta$ as a small perturbation on $\varrho_2$.

To this end, we consider isotropic states which are convex mixtures of a maximally entangled state $\vert\psi\rangle$ and the maximally mixed state $\id/D$. The condition for the separability of isotropic states is known \cite{HoHo99} and, in particular, the state
\begin{equation*}
\varrho_2 = \frac{1}{d+1} \vert \psi\rangle\langle\psi \vert + \frac{1}{d(d+1)}\id
\end{equation*}
lies on the boundary of the set of separable states. 
This state is therefore separable for any maximally entangled $\vert\psi\rangle = \frac{1}{\sqrt{d}}\sum_j \vert \phi_j\rangle\otimes\vert \varphi_j\rangle$, and the problem reduces to showing that for any  $\Delta$ the bases $\{\vert\phi_i\rangle\}_{i=1}^d$ and $\{\vert\varphi_j\rangle\}_{j=1}^d$ can be chosen in such a way that the perturbed state $\varrho_1$ is entangled. This will be done using the NPT criterion which will then immediately also give us the proof of the fact that distinguishing PPT states from NPT states requires informational completeness. The geometric idea behind the proof is presented in Fig.~\ref{fig:proof_method}.

\begin{figure}
\centering
            \includegraphics[width=9cm]{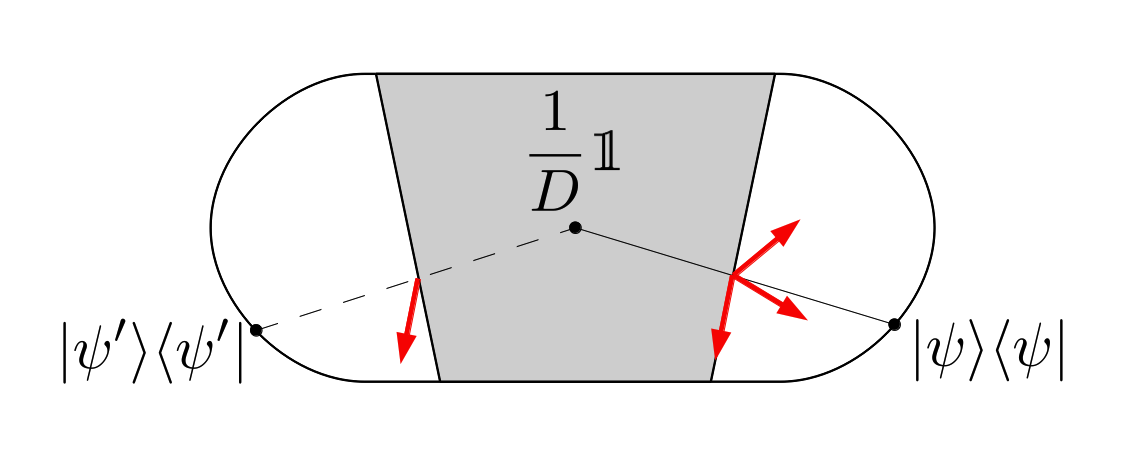} 
           \caption{The geometric idea behind the proof that entanglement detection requires an informationally complete measurement. For any direction (red arrow) we choose an isotropic state from the boundary of the set of separable states, and check if the perturbed state is entangled. If not, then we choose a different isotropic state. For any direction there always exists a separable state such that the perturbed state is entangled.}
    \label{fig:proof_method}
\end{figure}

By exploiting  the local unitary equivalence  of maximally entangled states, $\vert\psi\rangle=(U\otimes V) \vert\psi_0\rangle$ where $\vert\psi_0\rangle = \frac{1}{\sqrt{d}}\sum_j \vert j\rangle\otimes\vert j\rangle$ is the canonical maximally entangled state, we can express the partial transpose  with respect to the second factor as 
\begin{equation*}
\varrho_1^\tau = \frac{1}{d(d+1)}(U\otimes\overline{V})(\id + \flip)(U \otimes \overline{V})^* + \lambda \Delta^\tau
\end{equation*}
where $\flip $ is the flip operator. Since $\Delta^\tau$ is traceless and therefore has a negative eigenvalue, and zero is an eigenvalue of $\id + \flip$, it may seem obvious that an appropriate choice of the unitary operators would yield a negative eigenvalue for $\varrho_1^\tau$. 
However, the restriction to \emph{local} unitaries makes the problem nontrivial. The construction of the appropriate unitaries, and thus the completion of the proof, is given in the Supplementary Material.

It is worth stressing that this result cannot be obtained with a single fixed isotropic state, but one really needs the freedom in choosing the unitaries. As a counterexample,  set $U=V=\id$ and consider the operator
$$
\Delta = \vert 1\rangle \langle 1\vert \otimes   \vert 1\rangle \langle 1\vert  -   \vert d\rangle \langle d\vert   \otimes \vert d\rangle \langle d\vert
$$
so that $\Delta^\tau = \Delta$.  Now $\vert 1\rangle \otimes \vert 1 \rangle$ and $\vert d \rangle\otimes \vert d\rangle$ are eigenvectors of $\id + \flip$ corresponding to the eigenvalue $2$, and $0$ and $2$ are the only eigenvalues of $\id + \flip$. Therefore, the partially transposed state $\varrho_1^\tau $ remains positive whenever $\vert\lambda \vert$ is small. As explained previously, this means geometrically that the boundary between PPT and NPT states is flat in some directions at the points corresponding to the isotropic states. 

The remaining cases of CQ and CC states proceeds similarly, although now the condition for a state to belong to the set of discordant or nonclassical states is somewhat simpler. Indeed, by writing a state $\varrho$ in the computational basis as
\begin{equation*}
\varrho =  \sum_{ijkl} \varrho_{ijkl}\vert i \rangle \langle j \vert \otimes \vert k \rangle \langle l \vert
\end{equation*}
and by summing over  $i$ and $j$ we get the expression
\begin{equation*}
\varrho = \sum_{kl} A_{kl}(\varrho) \otimes \vert k\rangle \langle l \vert.
\end{equation*}
If the state is CQ, then the operators $A_{kl} (\varrho)$ are normal and commute \cite{GuCaCh12}. Similarly, by summing over $k$ and $l$ we get a family of operators $B_{ij} (\varrho)$, whose normality and commutativity is a necessary condition in order for the state to be QC. As CC states are those which are both CQ and QC, both families need to be checked in that case.

We show in the Supplementary Material that the state $\varrho_1 = \varrho_2 + \lambda^{-1} \Delta$ is CQ for all CQ states $\varrho_2$ and scalars $\lambda$ if and only if $\Delta = \id\otimes\Xi$ for some traceless Hermitian $\Xi$. This tells us two things: firstly, it is possible to decide if a state is CQ without identifying the state, and secondly, that the minimal number of measurement outcomes which is needed for this task is $d^4 - d^2 + 1$. The latter claim follows from Eqn.~\eqref{eqn:dim} and the fact that the dimension of the space of operators of the form $\id\otimes\Xi$ is $d^2-1$.

Quite interestingly, the remaining case of CC states again requires informational completeness.  This may seem counterintuitive since CC states are a subset of CQ states. However, one should keep in mind that their complements behave in the opposite way, namely, the discordant states are a subset of the nonclassical states. 

As a final observation we note that our result regarding quantum discord only works when the party possessing the ``quantum'' part has been fixed (i.e., we consider either CQ or QC states as our definition of the zero discord states). If we instead don't care about this and want to verify if a state belongs to the union of CQ and QC states, then we are again facing the necessity of an informationally complete measurement.

\section*{Discussion}\label{sec:discussion}

We have established that in order for a measurement to be able to decide the entanglement of  an unknown bipartite state, it is necessary for it to be informationally complete, thus allowing the full reconstruction of the unknown state. Since the number of measurement outcomes for such a measurement scales quadratically in the dimension of the system, it is clear that entanglement detection by measuring one single copy of the state is not an economical approach to the problem. Moreover, since the membership problem underlying entanglement detection is a simple yes-no question, one might hope to answer it more directly.

To establish this goal, the method of collective measurements has been brought to the context of entanglement detection. In a collective measurement, a finite number $N$ of identical systems are first prepared resulting in a factorized state $\varrho^{\otimes N}$, after which a global measurement (a POVM on the tensor product Hilbert space) is performed on the composite system. The benefit of going collective is that it allows one to measure polynomial functions of the  initial single-copy state \cite{Br04}, rather than just linear functions as in the case of standard measurements.  

An important result in  this line of research was the work of Augusiak {\em et al.} \cite{AuDeHo08}, where it was shown that the entanglement content of a two qubit state can be measured with a single binary measurement on four copies of the state. 
Up to our knowledge, such simple schemes are not known for higher dimensional systems but in principle collective measurements may be the way to go also in the case of more complex systems. 
In fact, it is known that the set of separable states is a semialgebraic set \cite{ChDo13} and as such, it is described by a finite set of polynomial inequalities $p_j(\varrho)\geq 0$. By going to collective measurements, these are transformed into linear inequalities whose validity can then be confirmed with simple binary measurements. However, it seems that for general systems there are no known bounds for the number of polynomials needed for the description of separable states, so it may happen that this method actually becomes intractable with increasing system size. The advantage of this approach is that it is computationally efficient, a fact that  serves as  motivation for  further studies of entanglement detection with collective measurements.

\section*{Acknowledgment}
JS and AT acknowledge financial support from the Italian Ministry of Education, University and Research (FIRB project RBFR10COAQ).


 \newpage

\onecolumngrid

\begin{center}
\large{{\bf Supplementary Material}}
\end{center}

\bigskip

\section*{Preliminaries}
\setcounter{equation}{0}
\renewcommand\theequation{S.\arabic{equation}}
For a given subset $\mathcal{P}$ of the state space, we want to know if for any nonzero traceless Hermitian operator $\Delta$,  there exists a state $\varrho\in\mathcal{P}$ such that 
\begin{equation}\label{eq:defk}
\kappa = \varrho + \lambda \Delta \in\mathcal{P}^c
\end{equation} 
for some real $\lambda$. Indeed, if this is the case, the membership problem for the set $\mathcal{P}$ can be solved only by informationally complete POVMs. Otherwise, we are going to determine a largest as possible subspace of traceless Hermitian operators $\Delta$ such that $\kappa\in\mathcal{P}$ for any real $\lambda$. Note that in Eqn.~\eqref{eq:defk} the complement is taken relative to the state space, so that not all $\lambda$'s are admissible, but only those ensuring the positivity of $\kappa$. We will address this issue by choosing $\varrho$ to be a full rank state, in which case it belongs to the interior of the state space, and it is then guaranteed that for small values of $\lambda$, $\kappa$ remains inside the state space. 

We first recall the definitions of the relevant subsets $\mathcal{P}$. The set of {\em separable states} consists of those states which are convex mixtures of product states:
\begin{equation*}
\varrho = \sum_i c_i \sigma_i\otimes \eta_i.
\end{equation*}
A separable state is called {\em classical-quantum} (CQ) if the states $\sigma_i$ are orthogonal one dimensional projections, and {\em quantum-classical} (QC) if $\eta_i$ are such. This means that the CQ and QC states are those which can be decomposed as
\begin{equation*}
\varrho_{\rm CQ} = \sum_i  c_i  \vert \phi_i \rangle \langle \phi_i \vert \otimes \eta_i, \qquad \varrho_{\rm QC}  = \sum_i c_i \sigma_i \otimes \vert \varphi_i\rangle \langle \varphi_i \vert
\end{equation*}
for some orthonormal bases $\{\vert \phi_i \rangle\}_{i=1}^d$ and $\{\vert \varphi_j \rangle\}_{j=1}^d$, respectively. If a state belongs to the intersection of CQ and QC states, it is called {\em classical-classical} (CC). These states admit a decomposition of the form
\begin{equation*}
\varrho_{\rm CC} = \sum_{ij} c_{ij} \vert \phi_i \rangle \langle \phi_i \vert \otimes \vert \varphi_j \rangle \langle \varphi_j\vert. 
\end{equation*}
Note that a product state $\varrho = \sigma \otimes \eta$ is CC by the spectral theorem applied to the states $\sigma$ and $\eta$. Moreover, for all states $\gamma$, the convex mixture $c_1 \varrho_{\rm CQ} + c_2 (\id/d)\otimes\gamma$ is CQ for all CQ states $\varrho_{\rm CQ}$ [respectively, $c_1 \varrho_{\rm QC} + c_2 \gamma\otimes(\id/d)$ is QC for all QC states $\varrho_{\rm QC}$].

In order to prove our results we need some useful criteria for deciding if a state belongs to the above subsets. In the case of separability we will use the PPT criterion \cite{Pe96S, HoHoHo96S} which states that the partial transpose of a separable state is a positive operator. 
In other words, a sufficient condition for entanglement is that the partial transpose has a negative eigenvalue. For the other types of correlations, we first write the state in the computational basis as
\begin{equation*}
\varrho = \sum_{ijkl} \varrho_{ijkl} \vert i \rangle \langle j \vert \otimes \vert k \rangle \langle l \vert .
\end{equation*}
After summing over $i$ and $j$, or $k$ and $l$, respectively, we obtain the two expressions
\begin{equation*}
\varrho  = \sum_{kl} A_{kl}(\varrho) \otimes \vert k \rangle \langle l \vert  = \sum_{ij} \vert i \rangle \langle j \vert \otimes B_{ij} (\varrho) .
\end{equation*}
If the state is classical-quantum, then the operators $A_{kl}(\varrho)$ are normal and commute, and if it is quantum-classical, then the same holds for $B_{ij} (\varrho)$ \cite{GuCaCh12S}. In particular, if the state is classical-classical then both families of operators are normal and commuting.

\section*{Separable states}

We start by proving that, if $\mathcal{P}$ is the set of separable states, the construction of Eqn.~\eqref{eq:defk} can always be done. 
As a candidate for $\varrho$, we fix a maximally entangled vector $\vert\psi\rangle$ and set 
\begin{equation*}
\varrho = \frac{1}{d+1} \vert \psi \rangle \langle \psi \vert + \frac{1}{d(d+1)} \id , 
\end{equation*}
so that $\varrho$ is a separable isotropic state which lies just on the boundary of the set of separable states \cite{HoHo99S}. 
This is a full rank state, so for a fixed traceless Hermitian operator $\Delta$ the operator  
$\kappa = \varrho  + \lambda \Delta$ 
is a valid state for small enough $\lambda$. 
The freedom that we still have is the choice of the maximally entangled vector $\vert\psi\rangle$. 
We will prove below that this can always be chosen in such a way that the partial transpose of $\kappa$ has a negative eigenvalue, and hence $\kappa$ is entangled. Before going into the details of the proof, we elaborate this idea a bit further. 

We use the fact that $\vert\psi\rangle$ can be obtained from the canonical maximally entangled vector $\vert\psi_0\rangle = 1/\sqrt{d}\sum_j \vert j\rangle \otimes \vert j \rangle$ by local unitaries: $\vert\psi\rangle = (U\otimes V)\vert\psi_0\rangle$. With this we can express the partial transpose of $\kappa$ with respect to the second factor as  
\begin{align*}
\kappa^\tau &= \frac{1}{d(d+1)} (U\otimes \overline{V}) (\id  + \flip) (U\otimes \overline{V})^*  + \lambda \Delta^\tau \\
&=  \frac{2}{d(d+1)}(U\otimes \overline{V})\left[ \frac{\id  + \flip}{2} + \lambda \frac{d(d+1)}{2}  (U\otimes \overline{V})^*\Delta^\tau  (U\otimes \overline{V})  \right] (U\otimes \overline{V}) ^*
\end{align*}
where $\flip $ is the flip operator $\flip (\vert\phi\rangle\otimes\vert\varphi\rangle )= \vert\varphi\rangle \otimes \vert\phi\rangle$.  Since the spectrum is invariant under the  unitaries, and the scaling of $\lambda$ as well as the complex conjugation of $V$ are irrelevant, it suffices to consider the operator
\begin{equation*}
\widetilde{\kappa}=\frac{\id  + \flip}{2} +  \lambda (U\otimes V)^*\Delta^\tau  (U\otimes V). 
\end{equation*}

Let now $W:\complex^2 \otimes \complex^2\to\complex^d \otimes \complex^d$ denote the canonical injection from $\complex^2 \otimes \complex^2$ into $\complex^d \otimes \complex^d$ (i.e., $W\vert i\rangle \otimes \vert j\rangle = \vert i\rangle \otimes \vert j\rangle$ for $i,j=1,2$). It is clearly sufficient to show that the operator 
\begin{equation*}
W^* \widetilde{\kappa}W=W^*\left[\frac{\id  + \flip}{2} +  \lambda (U\otimes V)^*\Delta^\tau  (U\otimes V)\right]W 
\end{equation*}
has a negative eigenvalue. By defining the orthonormal basis 
\begin{align*}
\vert f_1 \rangle & = \vert 1\rangle\otimes \vert 1\rangle & \vert f_2 \rangle & = \vert 2\rangle\otimes \vert 2\rangle \\
\vert f_3 \rangle & = \frac{1}{\sqrt{2}}(\vert 1\rangle\otimes \vert 2\rangle + \vert 2\rangle\otimes \vert 1\rangle) & \vert f_4 \rangle & = \frac{1}{\sqrt{2}}(\vert 1\rangle\otimes \vert 2\rangle - \vert 2\rangle\otimes \vert 1\rangle) \,.
\end{align*}
of $\complex^2 \otimes \complex^2$, we find that the matrix representation of $W^* \widetilde{\kappa}W$ with respect to this basis is 
\begin{equation}\label{eqn:matrix_form}
W^* \widetilde{\kappa}W = \left(\begin{array}{cc}
\id_{\C^3} & 0 \\
0 & 0
\end{array}\right)
+ \lam
\left(\begin{array}{cc}
A & a \\
a^* & \alpha
\end{array}\right)
\end{equation}
for some $3\times 3$ Hermitian matrix $A$, column vector $a\in\C^3$ and scalar $\alpha\in\R$. We now need the following Lemma.

\begin{lemma}\label{lemma:unitaries}
With the above notations, $U$ and $V$ can always be chosen in such a way that $a$ or $\alpha$ is nonzero.
\end{lemma}
\begin{proof}
Since $\Delta$ is nonzero, there exist $p,q,r,s\in\{1,2,\ldots,d\}$ for which
\begin{equation*}
\langle p \vert \otimes \langle q\vert  \Delta^\tau \vert r\rangle \otimes \vert s\rangle \neq 0.
\end{equation*}
Let $U$ and $V_0$ be two unitary operators on $\C^d$ such that $\{\vert p\rangle ,\vert r\rangle\} \subseteq U\complex^2$ and $\{ \vert q\rangle, \vert s\rangle\} \subseteq V_0\complex^2$. 
Then the previous inequality implies that
\begin{equation*}
\widetilde{\Delta} = W^*(U\otimes V_0)^*\Delta^\tau(U\otimes V_0)W \neq 0
\end{equation*}
since $\{\vert p\rangle\otimes \vert q\rangle , \vert r\rangle \otimes \vert s\rangle \} \subseteq \ran (U\otimes V_0)W$. We can now assume that $\widetilde{\Delta}_{4\,i} = 0$ for all $i\in\{1,2,3,4\}$, where the matrix elements are with respect to the basis $\{\vert f_j \rangle\}_{j=1}^4$, as otherwise the lemma is proved with $V= V_0 $.\\
If $T$ is any unitary operator on $\C^d$ which can be written in the block form as
$$
T =
\left(\begin{array}{cc}
T' & 0 \\
0 & \id_{\C^{d-2}}
\end{array}\right)
\qquad \text{with} \qquad
T' =
\left(\begin{array}{cc}
t_{1\,1} & t_{1\,2} \\
t_{2\,1} & t_{2\,2}
\end{array}\right)
$$
then $(\id\otimes T)WW^* = WW^*(\id\otimes T)$ and hence
\begin{align}
 W^*(U\otimes V_0T)^*\Delta^\tau(U\otimes V_0T)W &= W^*(\id\otimes T)^*W W^*(U\otimes V_0)^*\Delta^\tau(U\otimes V_0)W W^*(\id\otimes T)W \nonumber\\
&  = W^*(\id\otimes T)^*W \widetilde{\Delta} W^*(\id\otimes T)W , \label{eq:casino}
\end{align}
where the matrix representation of $W^*(\id\otimes T)W$ with respect to the basis $\{\vert f_j \rangle\}_{j=1}^4$ is given by 
\begin{equation}\label{eq:W*(1T)W}
W^*(\id\otimes T)W = \frac{1}{2}\,
\left(\begin{array}{cccc}
2 t_{1\,1} & 0 & \sqrt{2} t_{1\,2} & \sqrt{2} t_{1\,2} \\
0 & 2 t_{2\,2} & \sqrt{2} t_{2\,1} & -\sqrt{2} t_{2\,1} \\
\sqrt{2} t_{2\,1} & \sqrt{2} t_{1\,2} & t_{1\,1} + t_{2\,2} &  t_{2\,2} - t_{1\,1} \\
\sqrt{2} t_{2\,1} & -\sqrt{2} t_{1\,2} & t_{2\,2} - t_{1\,1} &  t_{1\,1} + t_{2\,2}
\end{array}\right)
\end{equation}
We now split the proof into two cases.

(1) Case $\widetilde{\Delta}_{3\,3} \neq 0$. Making the choice
$$
T \equiv T_0 =
\left(\begin{array}{cc}
T'_0 & 0 \\
0 & \id_{\C^{d-2}}
\end{array}\right)
\qquad \text{with} \qquad
T'_0 =
\left(\begin{array}{cc}
1 & 0 \\
0 & -1
\end{array}\right)
$$
Eqns.~\eqref{eq:casino} and \eqref{eq:W*(1T)W} yield
$$
[W^*(U\otimes V_0T_0)^*\Delta^\tau(U\otimes V_0T_0)W]_{4\,4} = \widetilde{\Delta}_{3\,3} \neq 0 .
$$
The claim then follows with $V = V_0T_0$.

(2) Case $\widetilde{\Delta}_{3\,3} = 0$. Picking the following two unitaries $T_+$, $T_-$
$$
T \equiv T_\pm =
\left(\begin{array}{cc}
T'_\pm & 0 \\
0 & \id_{\C^{d-2}}
\end{array}\right)
\qquad \text{with} \qquad
T'_\pm =
\left(\begin{array}{cc}
0 & \e^{\pm i \frac{\pi}{4}} \\
\e^{\mp i \frac{\pi}{4}} & 0
\end{array}\right)
$$
we have
\begin{align*}
[(U\otimes V_0T_{\pm})^*\Delta^\tau(U\otimes V_0T_{\pm})]_{4\,1} = -\frac{\widetilde{\Delta}_{2\,3} \pm i \widetilde{\Delta}_{1\,3}}{2} \\
[(U\otimes V_0T_{\pm})^*\Delta^\tau(U\otimes V_0T_{\pm})]_{4\,2} = \frac{\widetilde{\Delta}_{1\,3} \mp i \widetilde{\Delta}_{2\,3}}{2}
\end{align*}
and
\begin{align*}
[(U\otimes V_0T_{\pm})^*\Delta^\tau(U\otimes V_0T_{\pm})]_{4\,3} & = \frac{\widetilde{\Delta}_{1\,1} - \widetilde{\Delta}_{2\,2}}{2} \mp i\Re\widetilde{\Delta}_{1\,2} \\
[(U\otimes V_0T_{\pm})^*\Delta^\tau(U\otimes V_0T_{\pm})]_{4\,4} & = \frac{\widetilde{\Delta}_{1\,1} + \widetilde{\Delta}_{2\,2}}{2} \mp \Im\widetilde{\Delta}_{1\,2} .
\end{align*}
Since not all matrix elements $\widetilde{\Delta}_{1\,1}$, $\widetilde{\Delta}_{2\,2}$, $\widetilde{\Delta}_{1\,2}$, $\widetilde{\Delta}_{1\,3}$ and $\widetilde{\Delta}_{2\,3}$ can be zero, we must have $[(U\otimes V_0T_+)^*\Delta^\tau( U\otimes V_0T_+)]_{4\,i} \neq 0$ or $[(U\otimes V_0T_-)^*\Delta^\tau(U\otimes V_0T_-)]_{4\,i} \neq 0$ for some $i\in\{1,2,3,4\}$. Hence the proposition follows by choosing $V = V_0T_+$ or $V = V_0T_-$.
\end{proof}

With this lemma, we can now prove the main result of this section.
\begin{proposition}
For any nonzero traceless Hermitian operator $\Delta$, there exist unitaries $U$ and $V$ and a real number $\lambda$ such that $\kappa$ is a state and $\kappa^\tau$ has a negative eigenvalue.
\end{proposition}
\begin{proof}
By our previous discussion, it is sufficient to show that for some choice of $U$ and $V$, the matrix $W^* \widetilde{\kappa}W$ in Eqn.~\eqref{eqn:matrix_form} has a negative eigenvalue for all small $\lambda$'s. To this end, let $U$ and $V$ be as in Lemma~\ref{lemma:unitaries} and  let $S$ be a $3\times 3$ unitary matrix such that $S^*AS = {\rm diag}(\mu_1,\mu_2,\mu_3)$. Then
$$
W^* \widetilde{\kappa}W
=\left(\begin{array}{cc}
\id_{\C^3} & 0 \\
0 & 0
\end{array}\right)
+ \lam
\left(\begin{array}{cc}
A & a \\
a^* & \alpha
\end{array}\right) =
\left(\begin{array}{cc}
S & 0 \\ 0 & 1
\end{array}\right)
\left(\begin{array}{cccc}
1+\lam\mu_1 & 0 & 0 & \lam b_1\\
0 & 1+\lam\mu_2 & 0 & \lam b_2\\
0 & 0 & 1+\lam\mu_3 & \lam b_3\\
\lam \overline{b_1} & \lam \overline{b_2} & \lam \overline{b_3} & \lam\alpha
\end{array}\right)
\left(\begin{array}{cc}
S^* & 0 \\ 0 & 1
\end{array}\right)
$$
where $b = S^* a$. Therefore, we can easily compute
\begin{align*}
 \dett (W^* \widetilde{\kappa}W)=\, & \lam\alpha (1+\lam\mu_1)(1+\lam\mu_2)(1+\lam\mu_3) - \lam^2[|b_1|^2(1+\lam\mu_2)(1+\lam\mu_3) \\
&+ |b_2|^2(1+\lam\mu_1)(1+\lam\mu_3)  + |b_3|^2(1+\lam\mu_1)(1+\lam\mu_2)] .
\end{align*}
If $\alpha\neq 0$, then $\dett(W^* \widetilde{\kappa}W) = \lam\alpha + \mathcal{O}(\lam^2) $ and hence for small nonzero $\lam$'s with the opposite sign of $\alpha$ the determinant is negative. This implies that $W^* \widetilde{\kappa}W$ must have a negative eigenvalue. If $\alpha = 0$, then we must have $b\neq 0$ in which case $  \dett (W^* \widetilde{\kappa}W) = -\lam^2 \no{b}^2 + \mathcal{O}(\lam^3)$ which is negative for small nonzero $\lam$'s. Therefore we again conclude that  $W^* \widetilde{\kappa}W$ has a negative eigenvalue. This completes the proof.
\end{proof}

\begin{corollary}
In order to determine if a state is  PPT or NPT, the measurement must be informationally complete.
\end{corollary}
\begin{corollary}
In order determine if a state is separable or entangled, the measurement must be informationally complete.
\end{corollary}

\section*{CQ and QC states}
We now move to the case of CQ and QC states. In other words the subset $\mathcal{P}$ is now the CQ states or the QC states. We present the proof only for the CQ case, the proof for QC being equivalent with the obvious change of the local systems. The following proposition rules out the necessity of informational completeness. 

\begin{proposition}\label{prop:IX}
If $\Delta = \id \otimes \Xi$ for some traceless Hermitian operator $\Xi$, then  $\kappa = \varrho_{\rm CQ} + \lambda \Delta$ is a CQ state for all CQ states $\varrho_{\rm CQ}$ and real numbers $\lambda$ such that $\kappa$ is a valid state. 
\end{proposition}
\begin{proof}
Using the decomposition $\varrho_{\rm CQ} = \sum_i \vert \phi_i \rangle \langle \phi_i \vert \otimes \eta_i$ for some orthonormal basis $\{\vert \phi_i \rangle\}_{i=1}^d$, we immediately have $\kappa = \sum_i \vert \phi_i \rangle \langle \phi_i \vert \otimes (\eta_i  + \lambda \Xi)$ which proves the claim.
\end{proof}

The rest of this section goes to showing that whenever $\Delta \neq \id \otimes \Xi$, we can find a CQ state such that $\kappa$ is not CQ. In order to get a grasp of this situation, we decompose $\Delta$ as we did with states:
\begin{equation*}
\Delta = \sum_{ijkl} \Delta_{ijkl} \vert i \rangle \langle j \vert \otimes \vert k \rangle \langle l \vert  =  \sum_{kl} A_{kl}(\Delta)\otimes \vert k \rangle \langle l \vert.
\end{equation*}
If all of the operators $A_{kl} (\Delta)$ are multiples of the identity, i.e., $A_{kl} (\Delta) =a_{kl} (\Delta)\id$, then $\Delta = \sum_{kl} a_{kl}(\Delta) \id \otimes\vert k \rangle \langle l \vert =\id \otimes \Xi$ where $\Xi = \sum_{kl} a_{kl}(\Delta)\vert k \rangle \langle l \vert $. 
Conversely, if $\Delta = \id \otimes \Xi$, then $\Delta = \sum_{kl} \Xi_{kl} \id \otimes \vert k \rangle \langle l \vert = \sum_{kl} A_{kl}(\Delta) \otimes\vert k \rangle \langle l \vert$ with $A_{kl}(\Delta) = \Xi_{kl} \id$. In other words, the operators $\Delta = \id \otimes \Xi$ are exactly those for which all of the operators $A_{kl} (\Delta)$ are multiples of the identity. In order to exploit this observation, we prove the  following simple lemma.
\begin{lemma}\label{lem:l1}
If $A$ is not a multiple of the identity, then there exists a state $\sigma$ such that $[A,\sigma]\neq 0$. 
\end{lemma}
\begin{proof}
Suppose  that $[A,\sigma] = 0$ for all states $\sigma$. Since for any traceless Hermitian $\Delta$, the operator $\sigma = \frac{1}{d} (\id + \lambda \Delta)$ is a state for a small enough nonzero real number $\lambda$, we have that $[A, \Delta] = 0$ for all $\Delta$. But since the traceless Hermitian matrices and the identity span the whole space of matrices, $A$ commutes with any matrix. Hence, $A $ must be a multiple of the identity. 
\end{proof}

We can now prove the following result which completes this section.
\begin{proposition}\label{prop:lcc2}
If $\Delta$ is a traceless Hermitian matrix such that $\Delta \neq\id \otimes \Xi$, then there exists a CQ state $\varrho_{\rm CQ}$ and $\lambda$ such that $\kappa$ is a state but not CQ. 
\end{proposition}
\begin{proof}
Since $\Delta\neq \id \otimes \Xi$,  there exist $p$ and $q$ such that  $A_{pq}(\Delta)$ is not a multiple of the identity. By Lemma \ref{lem:l1} there is a state $\sigma$ such that the commutator $[A_{pq}(\Delta),\sigma]\neq 0$. Pick $t\neq p$, and define the state
\begin{equation}\label{eqn:cc_state}
\varrho_{\rm CQ} = \frac{1}{2}\sigma\otimes\vert t\rangle \langle t\vert + \frac{1}{2d^2}\id\otimes\id 
\end{equation}
which is clearly CQ (actually, it is CC). Furthermore,  since $\varrho_{\rm CQ}$ has full rank, there is an $\epsilon > 0$ such that $\kappa=\varrho_{\rm CQ}+\lam\Delta$ is a valid state for all $|\lam|<\epsilon$. Moreover,
\begin{equation*}
[A_{pq}(\varrho_{\rm CQ}+\lam\Delta) \,,\, A_{tt}(\varrho_{\rm CQ}+\lam\Delta)] = \frac{\lam}{2} [A_{pq}(\Delta) \,,\, \sigma] + \lam^2 [A_{pq}(\Delta) \,,\, A_{tt}(\Delta)]
\end{equation*}
because $A_{rs}(\varrho_{\rm CQ}) = \frac{1}{2} \delta_{rt}\delta_{st} \sigma + \frac{1}{2d^2}\delta_{rs}\id$ (recall that $t\neq p$). Choosing $\lam$ small enough, the right hand side of this equation can be made nonzero, and hence $\varrho_{\rm CQ} + \lam\Delta $ is not CQ. 
\end{proof}

\begin{corollary}
In order to determine if a state is CQ [resp., QC], the measurement does not have to be informationally complete, and the minimal number of measurement outcomes that is needed for this task is $d^4-d^2+1$.
\end{corollary}
\begin{proof}
By Propositions \ref{prop:IX} and \ref{prop:lcc2}, the linear space $\{\id\otimes\Xi \mid \text{$\Xi$ is traceless and Hermitian}\}$ coincides with the set $\mathcal{X}$ of traceless Hermitian operators $\Delta$ such that the state $\kappa = \varrho_{\rm CQ} + \lambda\Delta$ is CQ for all CQ states $\varrho_{\rm CQ}$ and valid $\lambda$'s. In particular, the set $\mathcal{X}$ is maximal among all the subspaces of traceless Hermitian operators sharing this property. By \cite[Proposition 1]{HeMaWo13S}, one can construct a POVM $E$ with $d^4-\dim\mathcal{X} = d^4-d^2+1$ outcomes such that $\mathcal{X}_E =\{ \Delta \mid \tr{\Delta E_j}=0\,\text{ for all }\, j \} \equiv \mathcal{X}$, and which is thus capable of determining if a state is CQ. Moreover, by maximality of $\mathcal{X}$, $d^4-d^2+1$ is actually the minimal number of outcomes required for this task.
\end{proof}

\section*{CC states}

We have already established all of the necessary ingredients needed for proving that informational completeness is required for determining if a state is CC or not. In fact, the key point is the observation that the state \eqref{eqn:cc_state} constructed in the proof of Proposition \ref{prop:lcc2} is actually CC. Let us now denote this state by $\varrho_{\rm CC}$. The proposition then directly implies that whenever $\Delta\neq \id \otimes \Xi$, there exists a $\lambda$ such that $\kappa = \varrho_{\rm CC} + \lambda \Delta$ is a state but not CQ. Since by definition CC states are both CQ and QC, we conclude that $\kappa$ is not CC. Due to the symmetric situation for QC states, we also have that whenever $\Delta\neq \Xi'\otimes \id$, there exists a $\lambda'$ such that $\kappa'= \varrho_{\rm CC} + \lambda'\Delta$ is a state but not CC. Now to close the reasoning we only need to note that the only possibility for having $\Delta = \id\otimes \Xi = \Xi' \otimes \id$ with $\tr{\Xi} = \tr{\Xi'} = 0$ is that $\Delta = 0$, which is not of concern here. We summarize this in the following proposition.
\begin{proposition}
For any nonzero traceless Hermitian $\Delta$, there exists a CC state $\varrho_{\rm CC}$ and a real number $\lambda$ such that $\kappa = \varrho_{\rm CC}   + \lambda \Delta$ is a valid state but not CC. 
\end{proposition}
\begin{corollary}
In order to determine if a state is CC or not, the measurement must be informationally complete.
\end{corollary}

\section*{CQ or QC states}

As the final case we consider the task of determining if a state is either a CQ or a QC state. This corresponds to the subset $\mathcal{P}$ being the union of the sets of CQ and QC states. Just as for CC states, we can exploit the fact that any nonzero traceless Hermitian operator satisfies either $\Delta \neq \id \otimes \Xi$ or $\Delta\neq \Xi' \otimes \id$. Using the symmetry between the CQ and QC states it is therefore sufficient to show that in the first case we can always find a CQ state $\varrho_{\rm CQ}$ such that $\kappa =\varrho_{\rm CQ} + \lambda \Delta$ is neither CQ nor QC. We first need a minor refinement of Lemma~\ref{lem:l1}.
\begin{lemma}\label{lem:l2}
If $A$ is  not a multiple of the identity, then there exists a state $\sigma$ such that $[A,\sigma]\neq 0$ and $\langle 1 \vert \sigma \vert 2\rangle  \neq 0$.
\end{lemma}
\begin{proof}
By Lemma~\ref{lem:l1} there exists a state $\sigma_0$ such that $[A,\sigma_0]\neq 0$. If $\langle 1 \vert  \sigma_0 \vert 2\rangle \neq 0$, then the second claim follows with $\sigma = \sigma_0$. If otherwise $\langle 1 \vert \sigma_0 \vert 2\rangle =0$, then for small nonzero positive $\mu $ the operator $\sigma = (1-\mu)\sigma_0 + \mu(\vert 1\rangle \langle 2\vert + \vert 2 \rangle \langle 1 \vert + \id)/d$ is still a state which does not commute with $A$, and $\langle 1 \vert \sigma \vert 2 \rangle  = \mu/d \neq 0$.
\end{proof}

\begin{proposition}
For any nonzero traceless Hermitian $\Delta$, there exists a CQ or QC state $\varrho$ and a real number $\lambda$ such that $\kappa = \varrho   + \lambda \Delta$ is a valid state but neither CQ nor QC.
\end{proposition}
\begin{proof}
As discussed above, we can restrict to the case $\Delta \neq \id \otimes \Xi$. Let $p$ and $q$ be such that $A_{pq}(\Delta)$ is not a multiple of the identity, and let $\sigma$ be a state such that $[A_{pq}(\Delta),\sigma]\neq 0$ and $\langle 1\vert \sigma \vert 2\rangle \neq 0$. Moreover, choose $t\neq p$, and let $\gamma$ be any state such that $[\vert t\rangle\langle t\vert,\gamma]\neq 0$. Such states $\sigma$ and $\gamma$ exist by Lemmas \ref{lem:l1} and \ref{lem:l2}. By possibly replacing $\gamma$ with the convex mixture $\gamma/2 + \id/(2d)$, we can assume that $\gamma$ has full rank. Then, the composite state
$$
\varrho_{\rm CQ} = \frac{1}{2}\sigma\otimes\vert t \rangle \langle t\vert  + \frac{1}{2d}\id\otimes\gamma
$$
has full rank and is CQ. We have
\begin{align*}
A_{rs}(\varrho_{\rm CQ}) & = \frac{1}{2} \delta_{rt}\delta_{st} \sigma + \frac{1}{2d} \langle r \vert \gamma \vert s\rangle \id \\
B_{rs}(\varrho_{\rm CQ}) & = \frac{1}{2} \langle r \vert \sigma \vert s\rangle \vert t\rangle\langle t \vert  + \frac{1}{2d} \delta_{rs} \gamma .
\end{align*}
Therefore, for small nonzero real numbers $\lam $ the operator $\varrho_{\rm CQ}+\lam\Delta$ is a state, and
\begin{align*}
 [A_{pq}(\varrho_{\rm CQ}+\lam\Delta)\,,\,A_{tt}(\varrho_{\rm CQ}+\lam\Delta)] =&\, \frac{\lam}{2} [A_{pq}(\Delta)\,,\,\sigma] + \lam^2 [A_{pq}(\Delta)\,,\,A_{tt}(\Delta)] \\
[B_{12}(\varrho_{\rm CQ}+\lam\Delta)\,,\,B_{11}(\varrho_{\rm CQ}+\lam\Delta)] =&\, \frac{1}{4d} \langle 1 \vert \sigma \vert 2\rangle  \, [\vert t\rangle \langle t\vert\,,\,\gamma] + \lam ([B_{12}(\Delta)\,,\,B_{11}(\varrho_{\rm CQ})] + [B_{12}(\varrho_{\rm CQ})\,,\,B_{11}(\Delta)]) \\
&+ \lam^2 [B_{12}(\Delta)\,,\,B_{11}(\Delta)] .
\end{align*}
Choosing $\lam$ small enough, both of the commutators can be made nonzero, and  hence $\varrho_{\rm CQ}+\lam\Delta$ is neither CQ nor QC.
\end{proof}

\begin{corollary}
In order to determine if a state is CQ or QC, the measurement must be informationally complete. 
\end{corollary}

\bibliographystyle{unsrt}

\begin{thebibliography}{10}

\bibitem{BeWi92}
C.H. Bennett and S.J. Wiesner.
\newblock Communication via one- and two-particle operators on
  {E}instein-{P}odolsky-{R}osen states.
\newblock {\em Phys. Rev. Lett.}, 69:2881--2884, 1992.

\bibitem{BeBrCrJoPeWo93}
C.H. Bennett, G.~Brassard, C.~Crepeau, R.~Jozsa, A.~Peres, and W.K. Wootters.
\newblock Teleporting an unknown quantum state via dual classical and
  {E}instein-{P}odolsky-{R}osen channels.
\newblock {\em Phys. Rev. Lett.}, 70:1895--1899, 1993.

\bibitem{RaBr01}
R.~Raussendorf and H.J. Briegel.
\newblock A one-way quantum computer.
\newblock {\em Phys. Rev. Lett.}, 86:5188--5191, 2001.

\bibitem{HoHoHoHo09}
R.~Horodecki, P.~Horodecki, M.~Horodecki, and K.~Horodecki.
\newblock Quantum entanglement.
\newblock {\em Rev. Mod. Phys.}, 81:865--942, 2009.

\bibitem{GuTo09}
O.~G\"{u}hne and G.~T\'{o}th.
\newblock Entanglement detection.
\newblock {\em Phys. Rep.}, 474:1--75, 2009.

\bibitem{HoHoHo96}
M.~Horodecki, P.~Horodecki, and R.~Horodecki.
\newblock Separability of mixed states: necessary and sufficient conditions.
\newblock {\em Phys. Lett. A}, 223:1--8, 1996.

\bibitem{ChSa14}
D.~Chru\'{s}ci\'{n}ski and G.~Sarbicki.
\newblock Entanglement witnesses: construction, analysis and classification.
\newblock {\em J. Phys. A: Math. Theor.}, 47:483001, 2014.

\bibitem{Bell64}
J.S. Bell.
\newblock On the {E}instein {P}odolsky {R}osen paradox.
\newblock {\em Physics}, 1:195--200, 1964.

\bibitem{ClHoShHo69}
J.F. Clauser, M.A. Horne, A.~Shimony, and R.A. Holt.
\newblock Proposed experiment to test local hidden-variable theories.
\newblock {\em Phys. Rev. Lett.}, 23:880--884, 1969.

\bibitem{Gu03}
L.~Gurvits.
\newblock Classical deterministic compexity of edmonds' problem and quantum
  entanglement.
\newblock In {\em Proceedings of the 35th Annual ACM Symposium on Theory of
  Computing, 9-11 June 2003, San Diego, CA, USA}, 2003.

\bibitem{Io07}
L.~M. Ioannou.
\newblock Computational complexity of the quantum separability problem.
\newblock {\em Quant. Inf. Comp.}, 7:335--370, 2007.

\bibitem{HeMaWo13}
T.~Heinosaari, L.~Mazzarella, and M.M. Wolf.
\newblock Quantum tomography under prior information.
\newblock {\em Comm. Math. Phys.}, 318:355--374, 2013.

\bibitem{ChDaJietal13}
J.~Chen, H.~Dawkins, Z.~Ji, N.~Johnston, D.~Kribs, F.~Shultz, and B.~Zeng.
\newblock Uniqueness of quantum states compatible with given measurement
  results.
\newblock {\em Phys. Rev. A}, 88:012109, 2013.

\bibitem{DaVeBr10}
B.~Daki\'{c}, V.~Vedral, and \v{C}. Brukner.
\newblock Necessary and sufficient condition for nonzero quantum discord.
\newblock {\em Phys. Rev. Lett.}, 105:190502, 2010.

\bibitem{MLQT12}
T.~Heinosaari and M.~Ziman.
\newblock {\em The {M}athematical {L}anguage of {Q}uantum {T}heory}.
\newblock Cambridge University Press, Cambridge, 2012.

\bibitem{Prugovecki77}
E.~Prugove\v{c}ki.
\newblock Information-theoretical aspects of quantum measurements.
\newblock {\em Int. J. Theor. Phys.}, 16:321--331, 1977.

\bibitem{GQS06}
I.~Bengtsson and K.~{\.Z}yczkowski.
\newblock {\em Geometry of quantum states}.
\newblock Cambridge University Press, Cambridge, 2006.

\bibitem{Werner89}
R.F. Werner.
\newblock Quantum states with {E}instein-{P}odolsky-{R}osen correlations
  admitting a hidden-variable model.
\newblock {\em Phys. Rev. A}, 40:4277--4281, 1989.

\bibitem{OlZu01}
H.~Ollivier and W.~H. Zurek.
\newblock Quantum discord: a measure of the quantumness of correlations.
\newblock {\em Phys. Rev. Lett.}, 88:017901, 2001.

\bibitem{CuVeDuCi03}
T.S. Cubitt, F.~Verstraete, W.~D\"ur, and J.I. Cirac.
\newblock Separable states can be used to distribute entanglement.
\newblock {\em Phys. Rev. Lett.}, 91:037902, 2003.

\bibitem{Gi-etal14}
D.~Girolami, A.~M. Souza, V.~Giovannetti, T.~Tufarelli, J.~G. Filgueiras, R.~S.
  Sarthour, D.~O. Soares-Pinto, I.~S. Oliveira, and G.~Adesso.
\newblock Quantum discord determines the interferometric power of quantum
  states.
\newblock {\em Phys. Rev. Lett.}, 112:210401, 2014.

\bibitem{Dakic-etal12}
B.~Dakic, Y.~O. Lipp, X.~Ma, M.~Ringbauer, S.~Kropatschek, S.~Barz, T.~Paterek,
  V.~Vedral, A.~Zeilinger, C.~Brukner, and P.~Walther.
\newblock Quantum discord as resource for remote state preparation.
\newblock {\em Nat. Phys.}, 8:666--670, 2012.

\bibitem{DiHoLeSmTe04}
D.P. DiVincenzo, M.~Horodecki, D.W. Leung, J.A. Smolin, and B.M. Terhal.
\newblock Locking classical correlations in quantum states.
\newblock {\em Phys. Rev. Lett.}, 92:067902, 2004.

\bibitem{Gu-etal12}
M.~Gu, H.~M. Chrzanowski, S.~M. Assad, T.~Symul, K.~Modi, T.~C. Ralph,
  V.~Vedral, and P.~K. Lam.
\newblock Observing the operational significance of discord consumption.
\newblock {\em Nat. Phys.}, 8:671--675, 2012.

\bibitem{HoHo99}
M.~Horodecki and P.~Horodecki.
\newblock Reduction criterion of separability and limits for a class of
  distillation protocols.
\newblock {\em Phys. Rev. A}, 59:4206--4216, 1999.

\bibitem{GuCaCh12}
Z.~Guoa, H.~Cao, and Z.~Chen.
\newblock Distinguishing classical correlations from quantum correlations.
\newblock {\em J. Phys. A: Math. Theor.}, 45:145301, 2012.

\bibitem{Br04}
T.~A. Brun.
\newblock Measuring polynomial functions of states.
\newblock {\em Quant. Inf. Comp.}, 4:401--408, 2004.

\bibitem{AuDeHo08}
R.~Augusiak, M.~Demianowicz, and P.~Horodecki.
\newblock Universal observable detecting all two-qubit entanglement and
  determinant-based entanglement tests.
\newblock {\em Phys. Rev. A}, 77:030301(R), 2008.

\bibitem{ChDo13}
L.~Chen and D.~\v{Z}. Dokovi\'{c}.
\newblock Dimensions, lengths, and separability in finite-dimensional quantum
  systems.
\newblock {\em J. Math. Phys.}, 54:022201, 2013.

\end{thebibliography}

\begin{thebibliography}{10}

\bibitem{Pe96S}
A.~Peres, 
\newblock Separability criterion for density matrices.
\newblock {\em Phys.~Rev.~Lett.}, 77:1413--1415, 1996.

\bibitem{HoHoHo96S}
M.~Horodecki, P.~Horodecki, and R.~Horodecki.
\newblock Separability of mixed states: necessary and sufficient conditions.
\newblock {\em Phys. Lett. A}, 223:1--8, 1996.

\bibitem{GuCaCh12S}
Z.~Guoa, H.~Cao, and Z.~Chen.
\newblock Distinguishing classical correlations from quantum correlations.
\newblock {\em J. Phys. A: Math. Theor.}, 45:145301, 2012.

\bibitem{HoHo99S}
M.~Horodecki and P.~Horodecki.
\newblock Reduction criterion of separability and limits for a class of distillation protocols.
\newblock {\em Phys. Rev. A}, 59:4206--4216, 1999.

\bibitem{HeMaWo13S}
T.~Heinosaari, L.~Mazzarella, and M.M. Wolf.
\newblock Quantum tomography under prior information.
\newblock {\em Comm. Math. Phys.}, 318:355--374, 2013.

\end{thebibliography}

\end{document}